\documentclass[12pt]{article}
\usepackage{enumerate}
\usepackage{amsfonts}
\usepackage{amsmath}
\usepackage{amssymb}
\usepackage{amsthm}
\usepackage{latexsym}
\usepackage{color}
\usepackage{graphicx}
\usepackage{wrapfig}
\usepackage{caption}
\usepackage{float}

\def\H{\mathcal{H}}

\def\S{\mathfrak{S}}
\def\P{\mathfrak{P}}

\def\T{\mathfrak{T}}

\def\N{\mathbb{N}}

\newcommand{\rank}{\mathrm{rank}}

\newcommand{\Tr}{\mathrm{Tr}}
\newcounter{defin}  \newcounter{lemma}  \newcounter{theorem}
\newcounter{proposition} \newcounter{corol}  \newcounter{remark} \newcounter{example}
\newenvironment{lemma}{\par\refstepcounter{lemma}     \textbf{Lemma \thelemma.} }{\rm\par}
\newenvironment{theorem}{\par\refstepcounter{theorem}     \textbf{Theorem \thetheorem.}\ }{\rm\par}
\newenvironment{proposition}{\par\refstepcounter{proposition}     \textbf{Proposition \theproposition.}\ }{\rm\par}
\newenvironment{corollary}{\par\refstepcounter{corol}     \textbf{Corollary \thecorol.} }{\rm\par}

\newenvironment{remark}{\par\refstepcounter{remark}     \textbf{Remark \theremark.}}{\rm\par}
\newenvironment{example}{\par\refstepcounter{example}     \textbf{Example \theexample.}}{\rm\par}

\begin{document}

\title{Partial majorization and Schur concave functions on the sets of quantum and classical states}
%state-dependent improvements

%\title{Continuity bound for the output entropy of a quantum channel with the input energy constraint and other results of these type}

%\title{Two families of tight semicontinuity bounds for the von Neumann entropy}

\author{M.E.~Shirokov\footnote{email:msh@mi.ras.ru}\\
Steklov Mathematical Institute, Moscow, Russia}
\date{}
\maketitle
%\vspace{50pt}
\begin{abstract}
We  construct for a Schur concave function $f$ on the set of quantum states a tight  upper bound on the difference  $f(\rho)-f(\sigma)$ for a quantum state $\rho$ with finite $f(\rho)$ and any quantum state $\sigma$ $m$-partially majorized by the state $\rho$ in the sense described in \cite{FCB}. We also  obtain a tight upper bound on this difference
under the additional condition $\frac{1}{2}\|\rho-\sigma\|_1\leq\varepsilon$ and find simple sufficient conditions for vanishing  this bound with $\,\min\{\varepsilon,1/m\}\to0\,$.

The obtained results are applied  to the von Neumann entropy. The concept of $\varepsilon$-sufficient majorization rank of a quantum state with finite entropy is introduced and a tight upper bound on this quantity is derived  and applied to the Gibbs states of a quantum oscillator.

We also show how the  obtained results can be reformulated for Schur concave functions on the set of probability distributions with a finite or countable set of outcomes.
\end{abstract}

\tableofcontents

\section{Introduction}

A function $f$ on the set $\S(\H)$ of states  of a quantum system described by Hilbert space $\H$ is called \emph{Schur concave} if
\begin{equation}\label{scr}
 f(\rho)\leq f(\sigma)
\end{equation}
for any quantum states   $\rho$ and $\sigma$ such that
\begin{equation}\label{m-cond+}
\sum_{i=1}^{k}p_i\geq\sum_{i=1}^{k}q_i
\end{equation}
for all natural $k$,  where  $\{p_i\}_{i=1}^{+\infty}$
and $\{q_i\}_{i=1}^{+\infty}$ are the sequences
of eigenvalues of the states $\rho$ and $\sigma$ arranged in the non-increasing order (taking the multiplicity into account) \cite{New,M-2,Bhatia}.\smallskip

A function $f$ on $\S(\H)$ is called \emph{Schur convex} if the function $-f$ is Schur concave.\smallskip

There are many functions on the set $\S(\H)$ used in quantum information theory which are either Schur concave or Schur convex. For example, the von Neumann, Renyi and Tsallis
entropies of a quantum state are  Schur concave functions (for any parameters of the latter two entropies). Note, at the same time, that the Renyi entropy of order $\,\alpha>1\,$ is not concave in the standard (Jensen) sense, so its Schur concavity can be treated as a particular substitution of the standard concavity.\smallskip

If inequality (\ref{m-cond+}) holds for any $k$ then, according to the modern terminology, we say that the state $\rho$ \emph{majorizes} the state $\sigma$ and write
$\,\rho\succ\sigma$ \cite{N&Ch,M-2}.\smallskip

If $\rho$ and $\sigma$ are infinite rank states then to verify the relation $\,\rho\succ\sigma$ we have to check the validity of an infinite number of inequalities. Naturally, the question arises what can be said about infinite rank states $\rho$ and $\sigma$ if inequality (\ref{m-cond+}) is valid only for  $k=1,2,...,m$. This leads us to the concept of \emph{$m$-partial majorization} for quantum states. According to the definition in \cite{FCB}  a state $\rho$  \emph{$m$-partially majorizes} a state $\sigma$ if inequality (\ref{m-cond+}) holds for  $\,k=1,2,..,m$. Denote this relation by $\,\rho\stackrel{\,m}{\succ}\sigma$. It is easy to see that this relation
is reflexive and transitive but it is not antisymmetric ($\,\rho\stackrel{\,m}{\succ}\sigma$ and $\,\sigma\stackrel{\,m}{\succ}\rho$ does not imply $\,\rho=\sigma$).\smallskip

If $\,\rank \rho=n<+\infty\,$  then the $(n-1)$-partial majorization of any state $\sigma$ by the state $\rho$ is equivalent to the standard majorization and, hence, implies inequality (\ref{scr}), but for each
natural $\,m<n-1\,$ it is easy to construct a state $\sigma$ such that $\,\sigma\stackrel{\,m}{\prec}\rho$, but the state $\sigma$ is not majorized by the state $\rho$.  The same claim holds
for any natural $m$ provided that $\rho$ is an infinite-rank state.

If $\,\rho\stackrel{\,m}{\succ}\sigma$ for some $m$  but $\,\rho\nsucc\sigma$ then we can not assert that (\ref{scr}) holds for a Schur concave function $f$. However,
we may try to use the relation $\,\rho\stackrel{\,m}{\succ}\sigma$ to obtain an upper bound on a possible violation of the inequality (\ref{scr})
which is characterised by the difference  $f(\rho)-f(\sigma)$. This task  is naturally extended to analysis of the maximal violation of inequality (\ref{scr}) for a given state $\rho$ and any state $\sigma$ $m$-partially majorized by $\rho$ and close to $\rho$ w.r.t. the trace norm distance.\smallskip

In \cite{FCB}, the above tasks are considered in the case when $f$ is the von Neumann entropy $S$ -- the most important Schur concave function used in quantum theory.
As a result, tight upper bounds on the difference $S(\rho)-S(\sigma)$ valid under the conditions
$$
\Tr H\rho\leq E,\quad \rho\stackrel{\,m}{\succ}\sigma\,\quad \textrm{and} \quad\textstyle\frac{1}{2}\|\rho-\sigma\|_1\leq\varepsilon
$$
are obtained and analysed (here $H$ is a positive operator treated as a Hamiltonian of a quantum system).\smallskip

In Section 3 of this article, we propose an universal way to obtain  upper bounds on
\begin{equation}\label{svp}
\sup\left\{f(\rho)-f(\sigma)\,\left|\, \sigma\stackrel{\,m}{\prec}\rho,\, \textstyle\frac{1}{2}\|\rho-\sigma\|_1\leq\varepsilon \right.\right\}
\end{equation}
for any Schur concave function $f$, a state $\rho$ with finite $f(\rho)$ and any $\varepsilon\in[0,1]$.\smallskip

In Section 4, we use  this technique to construct a tight upper bound on the supremum in (\ref{svp}) depending on the spectrum of $\,\rho\,$
and  find simple sufficient conditions for vanishing  this bound with $$\min\left\{\varepsilon,\frac{1}{m}\right\}\to0.$$
The simplest of these conditions (the lower semicontinuity of the function $f$ on $\S(\H)$) holds for the basic Schur concave functions used in quantum information theory, in particular, it holds for the von Neumann, Renyi and Tsallis entropies of a quantum state.\smallskip

In Section 5, we apply the obtained general results to the von Neumann entropy and use the tight upper bound on
$$
\sup\left\{S(\rho)-S(\sigma)\,\left|\, \sigma\stackrel{\,m}{\prec}\rho \right.\right\}
$$
to derive a tight upper bound on the \emph{$\varepsilon$-sufficient majorization rank} of a state $\rho$ with finite von Neumann entropy defined as the
minimal $m$ such that
$$
\sup\left\{\frac{S(\rho)-S(\sigma)}{S(\rho)}\,\left|\,\sigma\stackrel{\,m-1}{\prec}\rho \right.\right\}\leq\varepsilon
$$
(the use of the $(m-1)$-partial majorization here implies the  coincidence of the\break $0$-sufficient majorization rank of a state with the ordinary rank of this state).
We use this upper bound to estimate the $\varepsilon$-sufficient majorization rank of the Gibbs state of a quantum oscillator with different mean number of quanta.\smallskip

In Section 6, it is shown how the  obtained results can be reformulated for Schur concave functions on the set of probability distributions with a finite or countable set of outcomes.

In Section 7, a brief overview of the main  results of the article is given.

\section{Preliminaries}

Throughout the article we assume that $\mathcal{H}$ is an infinite-dimensional separable Hilbert space and $\mathfrak{T}( \mathcal{H})$ is the
Banach space of all trace-class operators on $\mathcal{H}$  with the trace norm $\|\!\cdot\!\|_1$. Write  $\mathfrak{T}_+(\mathcal{H})$ for the positive cone in $\mathfrak{T}( \mathcal{H})$. Denote the set of quantum states (operators
in $\mathfrak{T}_+(\mathcal{H})$ with unit trace) by $\mathfrak{S}(\mathcal{H})$ \cite{N&Ch,H-SCI,Wilde}.

%Write $I_{\mathcal{H}}$ for the unit operator on a Hilbert space
%$\mathcal{H}$ and $\id_{\mathcal{\H}}$ for the identity
%transformation of the Banach space $\mathfrak{T}(\mathcal{H})$.

The closed subspace spanned by the eigenvectors of an operator $\rho\in\T_{+}(\H)$ corresponding to its positive eigenvalues is called the \emph{support} of $\rho$ and denoted by $\mathrm{supp}\rho$.  The \emph{rank} $\rank\rho$ of $\rho$ is the dimension of $\mathrm{supp}\rho$.

We will essentially use the Mirsky inequality
\begin{equation}\label{Mirsky-ineq+}
  \sum_{i=1}^{+\infty}|p_i-q_i|\leq \|\rho-\sigma\|_1
\end{equation}
valid for any states $\rho$ and $\sigma$ in $\S(\H)$, where  $\{p_i\}_{i=1}^{+\infty}$
and $\{q_i\}_{i=1}^{+\infty}$ are the sequences
of eigenvalues of $\rho$ and $\sigma$ arranged in the non-increasing order (taking the multiplicity into account) \cite{Mirsky,Mirsky-rr}.

The \emph{von Neumann entropy} of a quantum state
$\rho \in \mathfrak{S}(\H)$ is defined by the expression
$\,S(\rho)=\Tr\eta(\rho)$, where  $\,\eta(x)=-x\ln x\,$ if $\,x>0\,$
and it is assumed that $\eta(0)=0$. The von Neumann entropy is a concave lower semicontinuous function on the set~$\mathfrak{S}(\H)$ taking values in~$[0,+\infty]$ \cite{H-SCI,L-2,W}.

We will use the  homogeneous extension $\widehat{S}$ of the von Neumann entropy to the positive cone $\T_+(\H)$ defined as
\begin{equation}\label{S-ext}
\widehat{S}(\rho)\doteq(\Tr\rho)S\!\left(\frac{\rho}{\Tr\rho}\right)=\Tr\eta(\rho)-\eta(\Tr\rho)
\end{equation}
for any nonzero operator $\rho$ in $\T_+(\H)$ and equal to $0$ at the zero operator \cite{L-2}.\smallskip

A state $\rho$ in $\S(\H)$ \emph{majorizes}
a state $\sigma$ in $\S(\H)$ if
\begin{equation*}%\label{m-cond++}
\sum_{i=1}^{k}p_i\geq\sum_{i=1}^{k}q_i,\quad \forall k,
\end{equation*}
where  $\{p_i\}_{i=1}^{+\infty}$
and $\{q_i\}_{i=1}^{+\infty}$ are the sequences
of eigenvalues of $\rho$ and $\sigma$ arranged in the non-increasing order (taking the multiplicity into account) \cite{New,M-2,Bhatia}.
This relation is usually denoted by $\,\rho\succ \sigma$.

A function $f$ on $\S(\H)$ is called \emph{Schur concave} if (cf.\cite{New,N&Ch,M-2,Bhatia,Grig})
$$
\rho\succ \sigma  \quad \Rightarrow\quad  f(\rho)\leq f(\sigma).
$$
A function $f$ on $\S(\H)$ is called \emph{Schur convex} if the function $-f$ is Schur concave.
\smallskip

The concepts of majorization of quantum states and Schur concave functions on the set of quantum states originate from the same
concepts for  probability distributions and functions on the set of probability distributions.\smallskip

A probability distribution  $\,\bar{p}=\{p_i\}_{i=1}^{n}$ with $\,n\leq+\infty\,$ outcomes majorizes a probability distribution $\,\bar{q}=\{q_i\}_{i=1}^{n}$ if
\begin{equation}\label{m-pd}
\sum_{i=1}^{k}p^{\downarrow}_i\geq\sum_{i=1}^{k}q^{\downarrow}_i\qquad  k=1,2,..,n,
\end{equation}
where $\{p_i^{\downarrow}\}_{i=1}^{n}$ and $\{q_i^{\downarrow}\}_{i=1}^{n}$
are the  probability distributions obtained from the distributions $\{p_i\}_{i=1}^{n}$ and $\{q_i\}_{i=1}^{n}$
by rearrangement in the non-increasing order. This relation is denoted by $\,\bar{p}\succ\bar{q}$.\smallskip

The Schur concave and Schur convex functions on the set $\P^n$ of all probability distributions with $\,n\leq+\infty\,$ outcomes
are defined using the partial order $"\succ"$ in the same way as  in the quantum case described before.\smallskip

\textbf{Note:} Throughout the article  we use the term  "$n$-tuple of numbers" in both cases $\,n\in\N\,$ and $\,n=+\infty\,$ simultaneously. In the second case we assume that it is a countable ordered subset of $\mathbb{R}$ (a sequence).

\section{Partial majorization and  Schur concavity}

According to the definition given in \cite{FCB} a quantum state $\rho$ $m$-partially majorizes
a quantum state $\sigma$ if
\begin{equation}\label{m-cond++}
\sum_{i=1}^{k}p_i\geq\sum_{i=1}^{k}q_i,\quad k=1,2,...,m,
\end{equation}
where  $\{p_i\}_{i=1}^{+\infty}$
and $\{q_i\}_{i=1}^{+\infty}$ are the sequences
of eigenvalues of the states $\rho$ and $\sigma$ arranged in the non-increasing order (taking the multiplicity into account). In this case we will write $\,\rho\stackrel{\,m}{\succ}\sigma$.

If $\,\rank \rho=n<+\infty\,$  then the $(n-1)$-partial majorization of any state $\sigma$ by the state $\rho$ is equivalent to the standard majorization but
for each natural $\,m<n-1\,$ it is easy to construct a state $\,\sigma\,$ such that $\,\rho\stackrel{\,m}{\succ}\sigma$ but $\,\rho\nsucc\sigma$.

If $\,\rank \rho=+\infty\,$  then
$$
\rho\succ\sigma \quad \Leftrightarrow \quad \rho\stackrel{\,m}{\succ}\sigma \quad \forall m\in \N.
$$

If $f$ is a Schur concave function and $\,\rho\stackrel{\,m}{\succ}\sigma$ for some $m$  but $\,\rho\nsucc\sigma$ then we can not assert that $f(\rho)\leq f(\sigma)$. However,
we may try to use the relation $\,\rho\stackrel{\,m}{\succ}\sigma$ to obtain an upper bound on a possible violation of the inequality $f(\rho)\leq f(\sigma)$
which is characterised by the difference  $f(\rho)-f(\sigma)$. It is natural also to try to improve this upper bound using the information about the distance between the states $\rho$ and $\sigma$.\smallskip

Our main technical tool for solving the above problems is the following

%The following proposition gives an universal way to solve the above mentioned tasks.

\begin{proposition}\label{pm} \emph{Let $f$ be a Schur concave function on the set $\,\S(\H)$ and $\rho$  be a state in $\,\S(\H)$  with  the spectral representation\footnote{Here and in what follows speaking about spectral representation (\ref{sprho}) we assume that $\{\varphi_i\}_{i=1}^{+\infty}$ is an orthonormal system of vectors in $\H$ and that some entries $p_i$ may be equal to zero (if $\rho$ is a finite rank state).}
\begin{equation}\label{sprho}
\rho=\sum_{i=1}^{+\infty} p_{i}|\varphi_i\rangle\langle \varphi_i|,\quad p_{i}\geq p_{i+1}\geq 0\; \textit{ for all }\; i,
\end{equation}
such that $f(\rho)$ is finite. Let
\begin{equation}\label{u-d}
U_\varepsilon(\rho)\doteq\left\{\sigma\in \S(\H)\,\left|\,\frac{1}{2}\|\rho-\sigma\|_1\leq\varepsilon\right.\right\}.
\end{equation}
For each  $\,m\,$ in $\,\N_0\doteq\{0\}\cup\N\,$ let
\begin{equation}\label{tm}
T_m(\rho)\doteq\left\{\left.\sum_{i=1}^{+\infty} q_{i}|\varphi_i\rangle\langle \varphi_i|\in\S(\H)\,\right|\,q_i= p_i,\, i=\overline{1,m},\;\,\textit{and}\;\; q_{m+1}\geq p_{m+1}\right\},
\end{equation}
where it is assumed that the condition $\,q_i= p_i,\, i=\overline{1,m}\,$ is omitted if $\;m=0$.}\footnote{We do not assume that $q_{i}\geq q_{i+1}$ for all $i$ in (\ref{tm}).}
\smallskip

\emph{Then
$$
\sup\left\{f(\rho)-f(\sigma)\,|\, \sigma \in U_\varepsilon(\rho)\right\}\leq \sup\left\{f(\rho)-f(\sigma)\,|\, \sigma \in T_0(\rho)\cap U_\varepsilon(\rho)\right\}
$$
and
$$
\sup\left\{f(\rho)-f(\sigma)\,\left|\, \sigma \in U_\varepsilon(\rho),\, \sigma\stackrel{\,m}{\prec}\rho \right.\right\}\leq \sup\left\{f(\rho)-f(\sigma)\,|\, \sigma \in T_{m}(\rho)\cap U_\varepsilon(\rho)\right\}
$$
for each natural $m$.}
\end{proposition}
\smallskip

Both claims of Proposition \ref{pm} follow from the corresponding  claims of Lemma \ref{pm+} below.\smallskip

\textbf{Note:} It is clear that
$$
T_0(\rho)\supseteq T_1(\rho)\supseteq T_2(\rho)\supseteq...\quad\textrm{ and }\quad \bigcap_{m=0}^{+\infty} T_m(\rho)=\{\rho\}.
$$
If $\,\rank\rho=n<+\infty\,$ then $\,T_m(\rho)=\{\rho\}\,$ for all $\,m\geq n-1$, if $\,\rank\rho=+\infty\,$ then $\,T_m(\rho)\neq\{\rho\}\,$ for all $m$.\smallskip\pagebreak

\begin{lemma}\label{pm+} \emph{Let
$\,\rho\,$ be a state in $\,\S(\H)$ with the spectral representation (\ref{sprho}). Let $\,T_m(\rho)$, $m=0,1,2,..,$ be the sets defined in (\ref{tm}).}\smallskip

\emph{If  $\,\sigma\,$ is an  arbitrary state in $\,\S(\H)$ then there exists a state $\sigma_*\in T_0(\rho)$ such that}
\begin{equation}\label{pm++}
\sigma_*\succ\sigma\quad\textit{ and } \quad \|\rho-\sigma_*\|_1\leq\|\rho-\sigma\|_1.
\end{equation}

\emph{If  $\,\sigma\,$ is a  state in $\,\S(\H)$ such that $\,\sigma\stackrel{\,m}{\prec}\rho\,$ for some natural $m$ then there exists a state $\sigma_*\in T_{m}(\rho)$ such that both relations in (\ref{pm++}) hold.}
\end{lemma}\smallskip

\emph{Proof.}  Assume that $\sigma$ is an arbitrary state in $\,\S(\H)$ with the spectral representation
$$
\sigma=\sum_{i=1}^{+\infty} q_{i}|\psi_i\rangle\langle \psi_i|
$$
such that $\,q_{i}\geq q_{i+1}\geq0\,$ for all $i$.

By applying Lemma 2A in \cite{FCB} to the probability distributions $\{p_{i}\}_{i=1}^{+\infty}$ and $\{q_{i}\}_{i=1}^{+\infty}$
we obtain a probability distribution $\{q^*_{i}\}_{i=1}^{+\infty}$ such that
\begin{equation*}%\label{3-c}
q^*_1\geq p_1\quad\textup{and}\quad  \sum_{i=1}^{+\infty}|p_i-q^*_{i}| \leq \sum_{i=1}^{+\infty}|p_i-q_{i}| \leq\|\rho-\sigma\|_1,
\end{equation*}
where the last inequality follows from the Mirsky inequality (\ref{Mirsky-ineq+}).
Moreover, the explicit construction of the distribution $\{q^*_{i}\}_{i=1}^{+\infty}$ given in the  proof Lemma 2A in \cite{FCB} shows that
\begin{equation}\label{mr=++}
\{q^*_{i}\}_{i=1}^{+\infty}\succ\{q_{i}\}_{i=1}^{+\infty},
\end{equation}
where $"\succ"$ is the standard  majorization order for probability distributions. It is clear  that the state
\begin{equation}\label{star}
\sigma_*=\sum_{i=1}^{+\infty} q^*_{i}|\varphi_i\rangle\langle \varphi_i|
\end{equation}
belongs to the set $T_0(\rho)$, $\,\sigma_*\succ\sigma\,$  and $\,\|\rho-\sigma_*\|_1 \leq\|\rho-\sigma\|_1$.

If $\,\sigma\,$ is a state such that $\,\sigma\stackrel{\,m}{\prec}\rho\,$ then we may repeat the above arguments by applying the construction from the proof Lemma 2B in \cite{FCB} to the probability distributions $\{p_{i}\}_{i=1}^{+\infty}$ and $\{q_{i}\}_{i=1}^{+\infty}$, since in this case all the inequalities in (\ref{m-cond++}) hold. As a result, we obtain a probability distribution $\{q^*_{i}\}_{i=1}^{+\infty}$ such that (\ref{mr=++}) holds,
\begin{equation*}%\label{3-c+}
q^*_i=p_i,\;\;\forall i=\overline{1,m},\quad q^*_{m+1}\geq p_{m+1} \quad\textup{and}\quad  \sum_{i=1}^{+\infty}|p_i-q^*_{i}| \leq \sum_{i=1}^{+\infty}|p_i-q_{i}| \leq\|\rho-\sigma\|_1,
\end{equation*}
where the last inequality follows from the Mirsky inequality (\ref{Mirsky-ineq+}).

In this case the state $\sigma_*$ defined in (\ref{star}) belongs to the set $T_{m}(\rho)$, $\,\sigma_*\succ\sigma\,$  and $\,\|\rho-\sigma_*\|_1 \leq\|\rho-\sigma\|_1$. $\Box$\smallskip

The condition of finiteness of $f(\rho)$ in Proposition \ref{pmc} can be omitted by reformulating
its claims. In fact, Lemma \ref{pm+} implies the following\smallskip

\begin{proposition}\label{pmc} \emph{Let $f$ be a Schur concave function on the set $\,\S(\H)$. Let
$\,\rho$ be a state in $\,\S(\H)$ with the spectral representation (\ref{sprho}). Then
$$
\inf\left\{f(\sigma)\,|\, \sigma \in U_\varepsilon(\rho)\right\}\geq \inf\left\{f(\sigma)\,|\, \sigma \in T_0(\rho)\cap U_\varepsilon(\rho)\right\}
$$
and
$$
\inf\left\{f(\sigma)\left|\, \sigma \in U_\varepsilon(\rho),\, \sigma\stackrel{\,m}{\prec}\rho \right.\right\}\geq \inf\left\{f(\sigma)\,|\, \sigma \in T_{m}(\rho)\cap U_\varepsilon(\rho)\right\}
$$
for each natural $\,m$, where $U_\varepsilon(\rho)$ and $\,T_m(\rho)$ are the sets defined in (\ref{u-d}) and (\ref{tm}).}
\end{proposition}

\section{Main theorem}

Let $\rho$  be a state in $\,\S(\H)$  with  the spectral representation
\begin{equation}\label{sprho+}
\rho=\sum_{i=1}^{n}p_{i}|\varphi_i\rangle\langle \varphi_i|,\quad p_i\geq p_{i+1}>0\;\; \forall i,\quad n\leq +\infty.
\end{equation}
Let $\,d_k\doteq 1-\sum_{i=1}^{k}p_k$, $\,k\in\N\cap[1,n+1)$.\smallskip

For each $\,m\in\{0\}\cup\N\,$ and $\,\varepsilon\in[0,1]\,$ we define the  state\footnote{The definition of the state $\rho_{m,\varepsilon}$ is motivated by the construction of the state $\rho_{*,\varepsilon}(\sigma)$ in \cite[Section 4.3]{H&D}: the state $\rho_{0,\varepsilon}$ coincides with the state $\rho_{*,\varepsilon}(\sigma)$ with $\sigma=\rho$ provided that $n<+\infty$.}
\begin{equation}\label{rho-me}
\rho_{m,\varepsilon}=\left\{\begin{array}{ll}
        \displaystyle\sum_{i=1}^{n}p^{m,\varepsilon}_{i}|\varphi_i\rangle\langle \varphi_i| &\quad\textrm{if}\;\;  m<n-1\;\textrm{ and }\;\varepsilon>0\\
        \rho &\quad\textrm{if either}\;\;  n<+\infty\;\; \textrm{and}\;\; m\geq n-1\;\textrm{ or }\;\varepsilon=0
\end{array}\right.\!,
\end{equation}
where $\{p^{m,\varepsilon}_{i}\}_{i=1}^{n}$ is the probability distribution defined as:
\begin{itemize}
  \item if $\,\varepsilon\geq d_{m+1}$ then
  \begin{equation}\label{f-1}
  p^{m,\varepsilon}_{i}=\left\{\begin{array}{ll}
        p_i &\textrm{if}\;\;  i\leq m\\
        d_{m} &\textrm{if}\;\;  i=m+1\\
        0 &\textrm{if}\;\;  i>m+1
       \end{array}\right.
  \end{equation}

  \item if $\,n<+\infty$ and $\,\varepsilon\leq p_n$ then
  \begin{equation}\label{f-2}
  p^{m,\varepsilon}_{i}=\left\{\begin{array}{ll}
        p_i &\textrm{if either}\;\;  i\leq m\;\;  \textrm{or} \;\;m+1<i<n \\
        p_{m+1}+\varepsilon &\textrm{if}\;\;  i=m+1\\
        p_n-\varepsilon &\textrm{if}\;\;  i=n
       \end{array}\right.
  \end{equation}
  \item if $\,\varepsilon<d_{m+1}$ and either $\,n=+\infty\,$ or $\,\varepsilon>p_n$ then
  \begin{equation}\label{f-3}
  p^{m,\varepsilon}_{i}=\left\{\begin{array}{ll}
        p_i &\textrm{if either}\;\;  i\leq m\;\; \textrm{or} \;\; m+1<i<\ell_\varepsilon\\
        p_{m+1}+\varepsilon &\textrm{if}\;\;  i=m+1\\
        p_{\ell_\varepsilon}-\varepsilon+d_{\ell_\varepsilon} &\textrm{if}\;\;  i=\ell_\varepsilon\\
        0 &\textrm{if}\;\;  i>\ell_\varepsilon
       \end{array}\right.,
  \end{equation}
where $\ell_\varepsilon=\min\{k\in \N\,|\, d_k\leq \varepsilon\}$ ($\ell_\varepsilon>m+1\,$ because $\,\varepsilon<d_{m+1}$).
\end{itemize}\smallskip

Now we can formulate the main result of this article.\smallskip

\begin{theorem}\label{main}
\emph{Let $f$ be a Schur concave function on the set $\,\S(\H)$. Let
$\rho$ be a state in $\,\S(\H)$ with the spectral representation (\ref{sprho+}) such that $f(\rho)$ is finite. Then
\begin{equation}\label{main+}
\sup\left\{f(\rho)-f(\sigma)\,\left|\, \sigma \in U_\varepsilon(\rho),\, \sigma\stackrel{\,m}{\prec}\rho \right.\right\}\leq f(\rho)-f(\rho_{m,\varepsilon})
\end{equation}
for each natural $\,m\,$ and $\,\varepsilon\in[0,1]$, where $\rho_{m,\varepsilon}$ is the state defined in (\ref{rho-me})  and $U_\varepsilon(\rho)$ is the $\varepsilon$-vicinity of the state $\rho$ defined in (\ref{u-d}).}

\emph{Inequality (\ref{main+}) is optimal in the following sense: for any $\,m\in\N\,$  and  $\,\varepsilon\in[0,1]\,$ there is a state $\rho$ such that
$\,\rho_{m,\varepsilon}\stackrel{\,m}{\prec}\rho\,$ and hence an equality holds in (\ref{main+}).}

\emph{The r.h.s. of (\ref{main+}) is a nondecreasing  function of $\,\varepsilon\,$ for each $\,m\in\N\,$ and a non-increasing function of $\,m\,$ for each $\,\varepsilon\in[0,1]$. It tends to zero as
\begin{equation}\label{con}
\min\left\{\varepsilon,\frac{1}{m}\right\}\to 0
\end{equation}
provided that one of following conditions hold:
\begin{enumerate}
  \item [$\rm a)$] the function $f$ is lower semicontinuous  on $\,\S(\H)$;
  \item [$\rm b)$] the limit relation
$$
\liminf_{n\to+\infty}f(\vartheta_n)\geq f(\vartheta_0)
$$
holds for any sequence $\{\vartheta_n\}\subset\S(\H)$ converging to a state $\vartheta_0\in\S(\H)$  with finite $f(\vartheta_0)$ such that
$\vartheta_n\succ\vartheta_0$ for all $\,n$.
\end{enumerate}}

\end{theorem}\smallskip

\begin{remark}\label{dh}
If we assume that  $\sigma\stackrel{0}{\prec}\rho$ holds trivially for all states $\rho$ and $\sigma$ then
inequality (\ref{main+}) remains valid for $\,m=0$. In this case it means that
\begin{equation*}%\label{main++}
\sup\left\{f(\rho)-f(\sigma)\,\left|\, \sigma\in U_\varepsilon(\rho)\right.\right\}\leq f(\rho)-f(\rho_{0,\varepsilon}).
\end{equation*}
In fact, the second claim of Corollary 3.2 in \cite{H&D} shows that \emph{an equality holds in this inequality} for any state $\,\rho\,$ and $\,\varepsilon\in[0,1]\,$
because the state $\,\rho_{0,\varepsilon}\,$ coincides with the state $\,\rho_{*,\varepsilon}(\sigma)\,$ with $\,\sigma=\rho\,$ constructed in \cite{H&D} (as mentioned before).\footnote{Strictly speaking,
the reference on Corollary 3.2 in \cite{H&D} is valid only in the case $\,n<+\infty$, but it is easy to upgrade
the arguments from \cite{H&D} to show that the state $\rho_{0,\varepsilon}$ majorizes
all the states in $U_\varepsilon(\rho)$ in the case  $\,n=+\infty\,$  as well.}
\end{remark}
\medskip

\emph{Proof.} Inequality (\ref{main+}) follows from Proposition \ref{pm} and Lemma \ref{lim} below, since
this lemma  implies that
$$
f(\sigma)\geq f(\rho_{m,\varepsilon})\quad \forall\sigma\in  T_{m}(\rho)\cap U_\varepsilon(\rho),
$$
where $U_\varepsilon(\rho)$ and $T_{m}(\rho)$ are the sets defined in (\ref{u-d}) and (\ref{tm}).

To prove the optimality of inequality (\ref{main+}) we have, for given  $m\in\N$  and  $\varepsilon\in(0,1]$, to construct a state $\rho$ with
the spectral representation (\ref{sprho+}) such that $\,\rho_{m,\varepsilon}\stackrel{\,m}{\prec}\rho\,$. If $\varepsilon<1$ then this can be done in two different ways:
\begin{itemize}
  \item to take the sequence $\{p_i\}_{i=1}^n$ such that $\,p_m\geq p_{m+1}+\varepsilon\,$ and $\,d_{m+1}>\varepsilon$, where
$$
d_{m+1}=1-p_1-p_2-...-p_{m+1};
$$
  \item to take the sequence $\{p_i\}_{i=1}^n$ such that $\,p_m\geq p_{m+1}+d_{m+1}\,$ and $\,d_{m+1}\leq\varepsilon$.
\end{itemize}
If $\,\varepsilon=1\,$ then the second of the above ways should be used.\medskip

To show that $f(\rho)-f(\rho_{m,\varepsilon})\leq f(\rho)-f(\rho_{m',\varepsilon'})\,$  for $m\geq m'$ and $\varepsilon\leq \varepsilon'$  note that
\begin{itemize}
  \item $f(\rho_{m,\varepsilon})=\inf\{f(\sigma)\,|\,\sigma\in T_m(\rho)\cap U_{\varepsilon}(\rho)\}$ for any  $m$ and $\varepsilon$ by the above proof;
  \item $T_m(\rho)\cap U_{\varepsilon}(\rho)\subseteq T_m(\rho')\cap U_{\varepsilon'}(\rho)\,$ for $m\geq m'$ and $\varepsilon\leq \varepsilon'$.
\end{itemize}

To prove the last claim of the theorem it suffices to note that
\begin{itemize}
  \item the state $\rho_{m,\varepsilon}$ tends to the state $\rho$ with the convergence (\ref{con}) w.r.t. the trace norm;
  \item the state $\rho_{m,\varepsilon}$ majorizes the state $\rho$ for any $\,m\in\N\,$  and  $\,\varepsilon\in(0,1]$;
\end{itemize}
$\Box$

\begin{lemma}\label{lim}  \emph{The state $\rho_{m,\varepsilon}$ belongs to the set $\,T_{m}(\rho)\cap U_\varepsilon(\rho)$  for each  $\,m\in\N\,$  and  $\,\varepsilon\in[0,1]\,$ and}
\begin{equation}\label{max}
\rho_{m,\varepsilon}\succ \sigma \quad \forall\sigma\in  T_{m}(\rho)\cap U_\varepsilon(\rho).
\end{equation}
\end{lemma}\smallskip

\emph{Proof.} If $\,n=\rank\rho<+\infty\,$ then we will assume that $\,m<n-1\,$, since otherwise $\,\rho_{m,\varepsilon}=\rho\,$ for any $\,\varepsilon\,$ and  $\,T_{m}(\rho)=\{\rho\}$.\smallskip

It is obvious that $\rho_{m,\varepsilon}\in T_{m}(\rho)$. It can be directly  verified that
$\,\frac{1}{2}\|\rho_{m,\varepsilon}-\rho\|_1=d_{m+1}\leq\varepsilon\,$ if $\,\rho_{m,\varepsilon}$ is defined by formula (\ref{f-1}) and that
$\,\frac{1}{2}\|\rho_{m,\varepsilon}-\rho\|_1=\varepsilon\,$ otherwise.

To prove (\ref{max}) we apply the arguments
used in the proof of Lemma 4.6 in \cite[Section 4.3]{H&D} with necessary modifications.

If $\varepsilon\geq d_{m+1}$ then $\{p^{m,\varepsilon}_{i}\}$ is defined by the formula  (\ref{f-1}) and , hence, (\ref{max}) directly follows from
the definition of the set $T_{m}(\rho)$ and the basic properties of the majorization order.

Assume that  $\{p^{m,\varepsilon}_{i}\}$ is defined by one of the formulae (\ref{f-2}) and (\ref{f-3}). If  $\{p^{m,\varepsilon}_{i}\}$ is defined by the formula (\ref{f-2})
then we set $\ell_\varepsilon=n$.

Let $\{q_i\}_{i=1}^{+\infty}$  be the spectrum of a state $\sigma\in  T_{m}(\rho)\cap U_\varepsilon(\rho)$
arranged in the non-increasing order. To prove that  $\rho_{m,\varepsilon}\succ \sigma$ it suffices to show that
\begin{equation}\label{max++}
 \sum_{i=1}^{k}p^{m,\varepsilon}_{i}\geq \sum_{i=1}^{k}q_{i}
\end{equation}
for all $\,k$, because inequality (\ref{max++}) implies the same inequality with the $n$-tuple $\{p^{m,\varepsilon}_{i}\}_{i=1}^{n}$ replaced by
its rearrangement in the non-increasing order.

Since $\,\|\rho-\sigma\|_1= \sum_{i=1}^{+\infty} |p_i-q_i|\leq 2\varepsilon$, while
$\{p_{i}\}_{i=1}^{n}$ and $\{q_{i}\}_{i=1}^{+\infty}$ are probability distributions, we have
$$
\sum_{i=1}^{+\infty} [p_i-q_i]_+=\sum_{i=1}^{+\infty} [p_i-q_i]_-\leq \varepsilon\qquad ([x]_{\pm}=\max\{\pm x, 0\}),
$$
where we assume that $p_i=0\,$ for $\,i>n$. So,
\begin{equation*}%\label{meq}
\sum_{i=1}^{k}(q_i-p_{i})\leq \varepsilon\quad \textrm{ and hence }\quad\sum_{i=1}^{k}q_{i}\leq \sum_{i=1}^{k}p_i+\varepsilon\quad  \forall k.
\end{equation*}
The last inequality implies (\ref{max++}) for all $k<l_\varepsilon$, since the construction of the distribution $\{p^{m,\varepsilon}_{i}\}_{i=1}^{n}$ implies
$$
\sum_{i=1}^{k}p^{m,\varepsilon}_{i}=\sum_{i=1}^{k}p_i+\varepsilon\quad \forall  k<l_\varepsilon.
$$

If $k\geq\ell_\varepsilon$ then the l.h.s. of (\ref{max++}) is equal to $1$. So, (\ref{max++}) holds trivially in this case. $\Box$\medskip

\begin{remark}\label{lim+} \textbf{(basic property of the state $\rho_{m,\varepsilon}$)}
Inequality (\ref{main+}) can be proved without using Proposition \ref{pm} by establishing the following property of the state $\rho_{m,\varepsilon}$:
$$
\sigma\stackrel{\,m}{\prec}\rho\quad\Rightarrow\quad  \sigma\prec\rho_{m,\varepsilon}\qquad \forall\sigma\in U_{\varepsilon}(\rho).
$$
If $\,\varepsilon=1\,$ then this property is directly verified. For $\,\varepsilon<1\,$ it can be proved by combining Lemmas \ref{pm+} and \ref{lim}: for any state $\,\sigma\in U_{\varepsilon}(\rho)\,$ such that $\,\sigma\stackrel{\,m}{\prec}\rho\,$ Lemma \ref{pm+} gives  an auxiliary state $\,\sigma_*\in T_m(\rho)\cap U_\varepsilon(\rho)\,$ such that $\,\sigma\prec\sigma_*$, while Lemma \ref{lim} shows that  $\,\sigma_*\prec\rho_{m,\varepsilon}$.
\end{remark}\medskip

By applying Theorem \ref{main} with $\varepsilon=1$ and denoting the state $\rho_{m,1}$ by $\rho_m$ we obtain the following\smallskip

\begin{corollary}\label{pm0} \emph{Let $\,m\,$ be a natural number and $\,\rho$  be a state in $\,\S(\H)$  with  the spectral representation (\ref{sprho+})
such that $f(\rho)$ is finite. Then
\begin{equation}\label{sub}
\sup\left\{f(\rho)-f(\sigma)\,\left|\, \sigma\stackrel{\,m}{\prec}\rho \right.\right\}\leq f(\rho)-f(\rho_m),
\end{equation}
where}
\begin{equation}\label{rho-m}
\rho_m\doteq\sum_{i=1}^{m} p_{i}|\varphi_i\rangle\langle \varphi_i|+(1-p_{1}-p_{2}-...-p_{m})|\varphi_{m+1}\rangle\langle \varphi_{m+1}|.
\end{equation}

\emph{Inequality (\ref{sub}) is tight: if the state $\rho$ is such that  $\,p_1+p_2+...+p_{m}\geq 1-p_m\,$  then $\,\rho_m\stackrel{\,m}{\prec}\rho\,$ and hence an equality
holds in (\ref{sub}).}\smallskip

\emph{The r.h.s. of (\ref{sub}) tends to zero as $\,m\to+\infty\,$ provided that the function $f$ satisfies one of the conditions $\,\rm a)\,$ and $\,\rm b)\,$ in the last claim of Theorem \ref{main}.}

\end{corollary}\smallskip

\begin{remark}\label{lim++} The upper bound (\ref{sub}) coincides with the upper bound  (\ref{main+}) in Theorem \ref{main} for any $\,\varepsilon\geq d_{m+1}\doteq 1-p_1-p_2-...-p_{m+1}\,$, since $\,\rho_m=\rho_{m,\varepsilon}\,$ for all such $\varepsilon$.
So, the information $\sigma \in U_\varepsilon(\rho)$  can be used to improve the upper bound (\ref{sub}) only if $\,\varepsilon<d_{m+1}$.
\end{remark}\smallskip

The  upper bound (\ref{sub}) is easily calculated and  can be applied to the von Neumann entropy $S$, the quantum Renyi entropy $R_\alpha$ and  the quantum Tsallis entropy $T_\alpha$ (in the role of $f$). Since all these entropies are lower semicontinuous functions, the last claim of Corollary \ref{pm0} shows that
\begin{equation}\label{E-r}
\lim_{m\to+\infty}\sup\left\{E(\rho)-E(\sigma)\,\left|\, \sigma\stackrel{\,m}{\prec}\rho \right.\right\}=0,\quad E=S, R_\alpha, T_\alpha,
\end{equation}
provided that $E(\rho)<+\infty$.\smallskip

\begin{example}\label{III} Consider the quantum Renyi entropy
\begin{equation*}%\label{R-e-d}
R_{\alpha}(\rho)=\frac{\ln\Tr\rho^{\alpha}}{1-\alpha},\quad\alpha>0,\quad \alpha\neq1,
\end{equation*}
of a state $\rho$ in $\,\S(\H)$  \cite{Wilde}. Assume that $\,\rho$  is a state with  the spectral representation (\ref{sprho+}) such that $R_{\alpha}(\rho)$ is finite. By setting $\,p_i=0\,$ for $\,i>n\,$ in the case $\,n<+\infty\,$  and using definition  (\ref{rho-m}) of the state $\,\rho_m\,$ we obtain
\begin{equation}\label{R-ineq}
\begin{array}{c}
\displaystyle R_{\alpha}(\rho)-R_{\alpha}(\rho_m)=
\displaystyle\frac{1}{1-\alpha}\left(\ln\left[\frac{\sum_{i=1}^{+\infty}p_i^\alpha}{\sum_{i=1}^{m}p_i^\alpha+\left[\sum_{i=m+1}^{+\infty}p_i\right]^\alpha}\right]\right)\\\\
\displaystyle=\left\{\begin{array}{ll}
       \frac{1}{1-\alpha}\left(\ln\left[1+\frac{\sum_{i=m+1}^{+\infty}p_i^\alpha-\left[\sum_{i=m+1}^{+\infty}p_i\right]^\alpha}{\sum_{i=1}^{m}p_i^\alpha
       +\left[\sum_{i=m+1}^{+\infty}p_i\right]^\alpha}\right]\right)  &\textrm{if}\quad \alpha\in(0,1)\\\\
       \frac{1}{\alpha-1}\left(\ln\left[1+\frac{\left[\sum_{i=m+1}^{+\infty}p_i\right]^\alpha-\sum_{i=m+1}^{+\infty}p_i^\alpha}{\sum_{i=1}^{+\infty}p_i^\alpha}\right]\right)   &\textrm{if}\quad \alpha\in(1,+\infty)
        \end{array}\right.\!.
\end{array}
\end{equation}
Thus, Corollary \ref{pm0} implies
$$
\! \sup\left\{R_{\alpha}(\rho)-R_{\alpha}(\sigma)\,\left|\, \sigma\stackrel{\,m}{\prec}\rho \right.\right\}\leq B\leq\left\{\begin{array}{ll}
       \!\frac{1}{1-\alpha}\left[\frac{\sum_{i=m+1}^{+\infty}p_i^\alpha-\left[\sum_{i=m+1}^{+\infty}p_i\right]^\alpha}{\sum_{i=1}^{m}p_i^\alpha
       +\left[\sum_{i=m+1}^{+\infty}p_i\right]^\alpha}\right]  &\textrm{if}\quad \alpha\in(0,1)\\\\
       \!\frac{1}{\alpha-1}\left[\frac{\left[\sum_{i=m+1}^{+\infty}p_i\right]^\alpha-\sum_{i=m+1}^{+\infty}p_i^\alpha}{\sum_{i=1}^{+\infty}p_i^\alpha}\right]   &\textrm{if}\quad \alpha\in(1,+\infty)
        \end{array}\right.\!,
$$
where $B$ is the r.h.s. of (\ref{R-ineq}). It is clear that  the r.h.s. of this inequality tends to zero as $\,m\to+\infty$ in both cases $\,\alpha\in(0,1)\,$ and $\,\alpha\in(1,+\infty)\,$ in accordance with (\ref{E-r}).

If the state $\rho$ is such that  $\,p_1+p_2+...+p_{m}\geq 1-p_m\,$  then an equality holds in the first of the above inequalities.
\end{example}\smallskip

The analogues bound can be easily written for the quantum Tsallis entropy $T_\alpha$ of any order $\alpha$.\smallskip

The case when $f$ is the von Neumann entropy is considered in Section 5 in  more detail.\smallskip

At the end of the section we present a version of Theorem \ref{main} formulated without
the condition of finiteness of $f(\rho)$. It is derived from Proposition \ref{pmc} by using Lemma \ref{lim} and the arguments from the proof of Theorem \ref{main}.\smallskip

\begin{theorem}\label{main-g}
\emph{Let $f$ be a Schur concave function on the set $\,\S(\H)$. Let
$\rho$ be a state in $\,\S(\H)$ with the spectral representation (\ref{sprho+}). Then
\begin{equation}\label{main+g}
\inf\left\{f(\sigma)\,\left|\, \sigma \in U_\varepsilon(\rho),\, \sigma\stackrel{\,m}{\prec}\rho \right.\right\}\geq f(\rho_{m,\varepsilon})
\end{equation}
for each natural $\,m$ and $\,\varepsilon\in[0,1]$, where $\rho_{m,\varepsilon}$ is the state defined in (\ref{rho-me}) and $U_\varepsilon(\rho)$ is the $\varepsilon$-vicinity of the state $\rho$ defined in (\ref{u-d}).}\smallskip

\emph{Inequality (\ref{main+g}) is optimal in the following sense: for any $m\in\N$  and  $\,\varepsilon\in[0,1]$ there is a state $\rho$ such that
$\,\rho_{m,\varepsilon}\stackrel{\,m}{\prec}\rho\,$ and hence an equality holds in (\ref{main+g}).}\smallskip

\emph{The r.h.s. of (\ref{main+g}) is a non-increasing  function of $\,\varepsilon$ for each $m\in\N$ and a nondecreasing  function of $\,m$ for each $\,\varepsilon\in[0,1]$. It tends to $\,f(\rho)\leq +\infty\,$ as
\begin{equation*}%\label{con-g}
\min\left\{\varepsilon,\frac{1}{m}\right\}\to 0
\end{equation*}
provided that one of following conditions hold:
\begin{enumerate}
  \item [$\rm a)$] the function $f$ is lower semicontinuous  on $\,\S(\H)$;
  \item [$\rm b)$] the limit relation
$$
\liminf_{n\to+\infty}f(\vartheta_n)\geq f(\vartheta_0)
$$
holds for any sequence $\{\vartheta_n\}\subset\S(\H)$ converging to a state $\vartheta_0\in\S(\H)$  with $\,f(\vartheta_0)\leq +\infty\,$ such that
$\vartheta_n\succ\vartheta_0$ for all $n$.
\end{enumerate}}
\end{theorem}

Inequality (\ref{main+g}) remains valid for $\,m=0\,$ if we assume that  $\sigma\stackrel{0}{\prec}\rho$ holds trivially for all states $\rho$ and $\sigma$.
In this case it means that
\begin{equation*}%\label{main++}
\inf\left\{f(\sigma)\,\left|\, \sigma\in U_\varepsilon(\rho)\right.\right\}\geq f(\rho_{0,\varepsilon}).
\end{equation*}
By the arguments in \cite{H&D} and the comments in Remark \ref{dh} \emph{an equality holds in this inequality} for any state $\,\rho\,$ and $\,\varepsilon\in[0,1]$.

\section{Application to the von Neumann entropy}

\subsection{General results}

In this section we apply the results of Section 4 to the von Neumann entropy -- the most important Schur concave function used in quantum theory.

Assume that $\rho$ is a state in $\,\S(\H)$ with the spectral representation (\ref{sprho+}). For any natural $\,m<n\doteq\rank\rho\,$  define the positive trace class operator
\begin{equation}\label{d-d+}
 \rho^{[m]}=\sum_{i=m+1}^{n} p_{i}|\varphi_i\rangle\langle \varphi_i|,
\end{equation}
i.e. $\rho^{[m]}$ is the operator obtained from the state $\rho$ by
removing the $m$-rank component corresponding to its $m$ maximal eigenvalues.\footnote{If the state $\rho$  has multiple eigenvalues then the orthonormal system $\{\varphi_i\}_{i=1}^{n}$ is not uniquely defined and hence there is an ambiguity in the definition of the operator $\rho^{[m]}$. However, the spectrum of $\rho^{[m]}$ is uniquely defined. So, dealing with the quantities depending on the spectrum of $\rho^{[m]}$ we may forget about this ambiguity.} If $\,n\doteq\rank\rho<+\infty\,$ then we assume that $\rho^{[m]}=0$ for any $m\geq n$.

Theorem \ref{main} implies (due to the Schur concavity and the lower semicontinuity of the von Neumann entropy) the following\smallskip

\begin{proposition}\label{main-S}
\emph{Let $\rho$ be a state in $\,\S(\H)$ with the spectral representation (\ref{sprho+}) such that $S(\rho)$ is finite.  Let $\,d_k\doteq1-p_1-p_2-...-p_k\,$ for $\,k\in\N\cap[1,n+1)$. Let  $\,m\in\N\,$ and $\,\varepsilon\in[0,1]\,$ be arbitrary and $\,\ell_{\varepsilon}\doteq\min\{k\in \N\,|\, d_k\leq \varepsilon\}$. Then
\begin{equation}\label{main+S}
\sup\left\{S(\rho)-S(\sigma)\,\left|\, \sigma \in U_\varepsilon(\rho),\, \sigma\stackrel{\,m}{\prec}\rho \right.\right\}\leq B(\rho,m,\varepsilon),
\end{equation}
where
$$
B(\rho,m,\varepsilon)\doteq\left\{\begin{array}{ll}
        \widehat{S}(\rho^{[m]})&\textrm{if}\;\;\varepsilon\geq d_{m+1}\\
        \Delta(\rho,m,\varepsilon)+\widehat{S}(\rho^{[\ell_\varepsilon-1]})&\textit{if}\;\; \varepsilon<d_{m+1}
        \end{array}\right.\!,
$$
$$
\Delta(\rho,m,\varepsilon)\doteq \eta(p_{m+1})+\eta(d_{\ell_\varepsilon-1})-\eta(p_{m+1}+\varepsilon)-\eta(d_{\ell_\varepsilon-1}-\varepsilon),
$$
$\rho^{[m]}$ and $\rho^{[\ell_\varepsilon-1]}$ are the operators defined according to the rule (\ref{d-d+}),
$\widehat{S}$ is the extension of the von Neumann entropy defined in (\ref{S-ext}) and $U_\varepsilon(\rho)$ is the $\varepsilon$-vicinity of the state $\rho$ defined in (\ref{u-d}).}\smallskip

\emph{Inequality (\ref{main+S}) is optimal in the following sense: for any $m\in\N$  and  $\varepsilon\in[0,1]$ there is a state $\rho$ such that
$\,\rho_{m,\varepsilon}\stackrel{\,m}{\prec}\rho\,$ and hence an equality holds in (\ref{main+S}).}\smallskip

\emph{$B(\rho,m,\varepsilon)$ is a nondecreasing function of $\,\varepsilon$ for each $m\in\N$ and a non-increasing function of $\,m$ for each $\,\varepsilon\in[0,1]$. Moreover,}
\begin{equation*}%\label{con-S}
B(\rho,m,\varepsilon)\rightarrow0\qquad  \textit{as}\qquad  \min\left\{\varepsilon,\frac{1}{m}\right\}\to 0.
\end{equation*}
\end{proposition}
\smallskip

\textbf{Note A:} It can be directly verified that  $B(\rho,m,\varepsilon)=S(\rho)-S(\rho_{m,\varepsilon})$.
\smallskip

\textbf{Note B:} By Remark \ref{dh} in Section 4 inequality (\ref{main+S})  remains valid for $\,m=0\,$ if we assume  that $\sigma\stackrel{0}{\prec}\rho$ holds trivially for all states $\rho$ and $\sigma$.
In this case it means that
\begin{equation*}%\label{main++}
\sup\left\{S(\rho)-S(\sigma)\,\left|\, \sigma\in U_\varepsilon(\rho)\right.\right\}\leq B(\rho,0,\varepsilon).
\end{equation*}
Moreover, \emph{an equality holds in this inequality} for any state $\,\rho\,$ and $\,\varepsilon\in[0,1]$ by the arguments from \cite{H&D} and the
comments in Remark \ref{dh}.\smallskip

\begin{example}\label{g-exam}
Assume that
\begin{equation}\label{G-s}
 \rho_N\doteq(1-q)\sum_{i=1}^{+\infty} q^{i-1}|\phi_i\rangle\langle \phi_i|,\qquad q=\frac{N}{N+1},
\end{equation}
is the Gibbs state of a quantum oscillator  corresponding to the mean number of quanta $N$,
where $\,\{\phi_i\}_{i=1}^{+\infty}\,$ is the Fock basis in $\H$ \cite[Ch.12]{H-SCI}.
Then
$$
d_k=q^k\quad\textrm{ and }\quad \widehat{S}(\rho_N^{[k]})=q^{k}\,\frac{h(q)}{1-q},\quad k=0,1,2,...\quad(\rho_N^{[0]}=\rho_N).
$$
Hence, the r.h.s. of (\ref{main+S}) is equal to
$$
B(\rho_N,m,\varepsilon)\doteq\left\{\begin{array}{ll}
        q^{m}S(\rho_N)&\textrm{if}\;\;\varepsilon\geq q^{m+1}\\
        \Delta(\rho_N,m,\varepsilon)+\varepsilon q^{-\{\log_q \varepsilon\}}S(\rho_N)&\textrm{if}\;\; \varepsilon<q^{m+1}
        \end{array}\right.\!,
$$
where
$$
S(\rho_N)=\frac{h(q)}{1-q},\qquad \{\log_q \varepsilon\}\;\textrm{ is the fractional part of }\log_q \varepsilon
$$
and
$$
\Delta(\rho_N,m,\varepsilon)\doteq \eta(q^m(1-q))+\eta(\varepsilon q^{-\{\log_q \varepsilon\}})-\eta(q^m(1-q)+\varepsilon)-\eta(\varepsilon (q^{-\{\log_q \varepsilon\}}-1)).
$$
\end{example}

In Figures 1 and 2, the plots of the functions $\,\varepsilon\mapsto B(\rho_N,m,\varepsilon)\,$ with $\,N=2\,$ and $\,N=10\,$ for different values of $\,m\,$ are shown along with
the plots of the function
$$
\varepsilon\,\mapsto\, \sup\left\{S(\rho_N)-S(\sigma)\,\left|\, \sigma \in U_\varepsilon(\rho_N)\right.\right\}=B(\rho_N,0,\varepsilon)
$$
(see the remark after Proposition \ref{main-S}). The plot of the latter function is marked by $m=0$. The wavy structure of these plots is related to the term $\,q^{-\{\log_q \varepsilon\}}$. \medskip

\begin{figure}[t]

\centering
\begin{center}

\includegraphics[scale=0.4, bb=400  450 500 550]{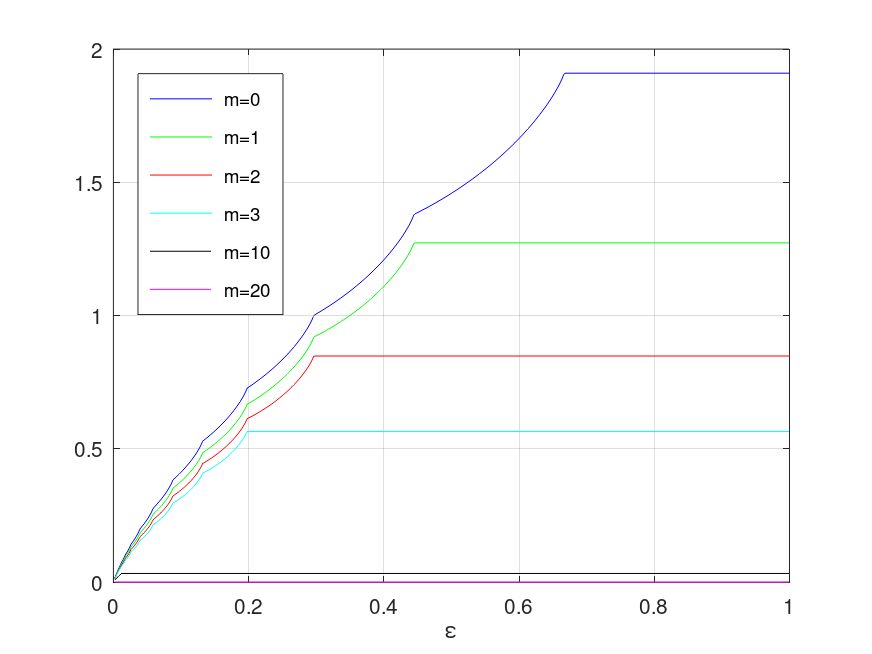}

\vspace{180pt}
\caption{The function $\,\varepsilon\mapsto B(\rho_N,m,\varepsilon)\,$ with $E=2$ and $m=0,1,2,3,10,20$.}
\end{center}

\label{Fig1}
\end{figure}

\begin{figure}[t]

\centering
\begin{center}

\includegraphics[scale=0.4, bb=400  450 500 550]{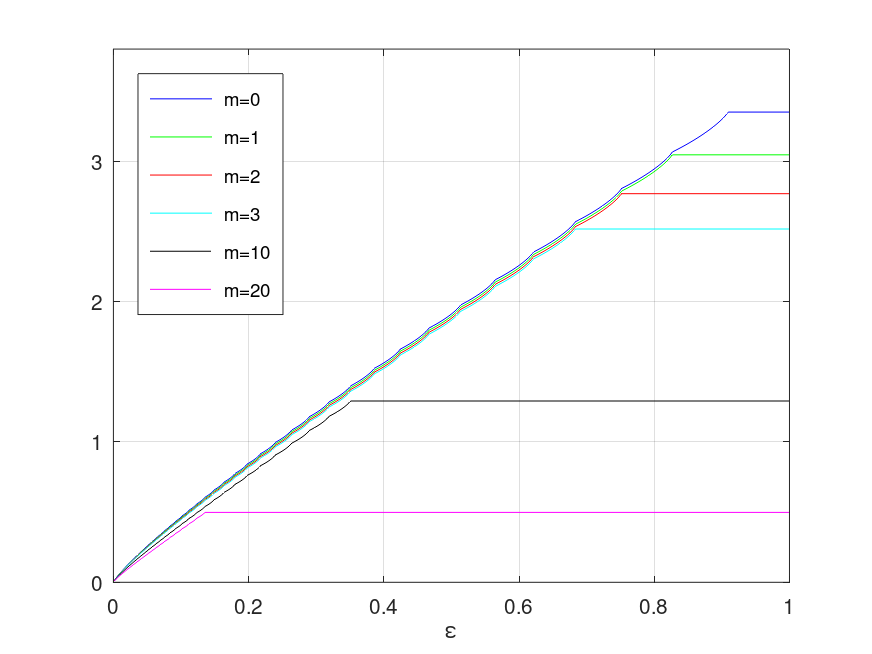}

\vspace{180pt}
\caption{The function $\,\varepsilon\mapsto B(\rho_N,m,\varepsilon)\,$ with $E=10$ and $m=0,1,2,3,10,20$.}
\end{center}

\label{Fig2}
\end{figure}

Denoting the state $\,\rho_{m,\varepsilon}$ with $\,\varepsilon=1\,$   we obtain the following
\smallskip

\begin{corollary}\label{pmS} \emph{Let $\,m\,$ be a natural number and $\,\rho$  be a state in $\,\S(\H)$  with  the spectral representation (\ref{sprho+})
such that $S(\rho)$ is finite. Then
\begin{equation}\label{sub+}
\sup\left\{S(\rho)-S(\sigma)\,\left|\, \sigma\stackrel{\,m}{\prec}\rho \right.\right\}\leq \widehat{S}(\rho^{[m]}),
\end{equation}
where $\rho^{[m]}$ is the operator defined in (\ref{d-d+}) and $\,\widehat{S}$ is the extension of the von Neumann entropy to the cone $\,\T_+(\H)$ defined in (\ref{S-ext}).}\smallskip

\emph{Inequality (\ref{sub}) is tight: if the state $\rho$ is such that  $\,p_1+p_2+...+p_{m}\geq 1-p_m\,$  then the state $\rho_m$ defined in (\ref{rho-m}) is $m$-partially majorized by the state $\rho$ and hence an equality
holds in (\ref{sub+}).}\smallskip

\emph{The r.h.s. of (\ref{sub+}) monotonously tends to zero as $\,m\to+\infty$.}
\end{corollary}

\subsection{On $\varepsilon$-sufficient majorization rank of a state}

If $\rho$ is a state of finite rank $n$ then the  $(n-1)$-partial majorization of a state $\sigma$ by the state $\rho$
implies the standard majorization of $\sigma$ by $\rho$, and, hence, the inequality
\begin{equation}\label{S-ineq+}
  S(\sigma)\geq S(\rho).
\end{equation}
If $\rho$ is an infinite rank state then the $m$-partial majorization of a state $\sigma$ by the state $\rho$ for
a given arbitrary  $m$ does not imply inequality (\ref{S-ineq+}). Nevertheless, if $\rho$ is a state with finite entropy then Corollary \ref{pmS} shows that
the $m$-partial majorization of a state $\sigma$ by the state $\rho$ implies that the difference $S(\rho)-S(\sigma)$ cannot exceed some bound (depending on $\rho$ and $m$), which tends to zero as $\,m\to+\infty$.

Motivating by this observation  introduce, for a given $\varepsilon\geq0$,  the following characteristic of a mixed  state $\rho$ in $\,\S(\H)$ with finite entropy
\begin{equation}\label{mr}
mr_\varepsilon(\rho)\doteq \inf\left\{\,m\in\N\;\left|\;\,\sup_{\sigma\stackrel{\,m}{\prec}\rho}\,\frac{S(\rho)-S(\sigma)}{S(\rho)}\,\leq \,\varepsilon\,\right.\right\}+1.
\end{equation}

We use such definition, since it
seems reasonable to characterize the degree of violation of the inequality (\ref{S-ineq+}) using the \emph{relative error} $\,\epsilon\doteq\frac{S(\rho)-S(\sigma)}{S(\rho)}\,$ rather than the value of $\,S(\rho)-S(\sigma)$.

\textbf{Note:} The last claim of Corollary \ref{pmS}   implies that $mr_\varepsilon(\rho)$ is a finite natural number for any state  $\rho$ with finite entropy and $\,\varepsilon>0$.\smallskip

It is natural  to call  $mr_\varepsilon(\rho)$ the \emph{$\varepsilon$-sufficient majorization rank of a state $\rho$}. It characterizes
the rate of decreasing of the spectrum of $\rho$.\smallskip

Note that definition (\ref{mr}) with $\varepsilon=0$ gives $\,mr_0(\rho)=\rank\rho\,$ for any finite rank mixed state $\rho$. If $\rho$ is an
infinite rank state then $\,mr_0(\rho)=+\infty$, since for any natural $m$ one can find a state $\sigma$ such that $\,\sigma\stackrel{\,m}{\prec}\rho\,$ and $\,S(\sigma)<S(\rho)$.\smallskip

Corollary \ref{pmS} implies the following\smallskip

\begin{corollary}\label{smr} \emph{Let $\rho$ be a mixed state in $\,\S(\H)$ with finite entropy. Then
\begin{equation}\label{smr+}
mr_\varepsilon(\rho)\leq \min\left\{m\in\N\,\left|\, \widehat{S}(\rho^{[m]})\leq \varepsilon S(\rho) \right.\right\}+1,
\end{equation}
where $\,\widehat{S}$ is the extension of the von Neumann entropy to the cone $\,\T_+(\H)$ defined in (\ref{S-ext}) and $\rho^{[m]}$ is the state defined in (\ref{d-d+}).}\smallskip

\emph{Inequality (\ref{smr+}) is tight: for any $\varepsilon\in[0,1]$ there is a state $\rho$ such that  an equality holds in (\ref{smr+}).}
\end{corollary}\smallskip

Since $\,\widehat{S}(\rho^{[m]})\,$ tends to zero as $\,m\to+\infty$ for any state with finite $S(\rho)$,  the r.h.s. of (\ref{smr+}) is
a well defined natural number. It will be denoted by $\widehat{mr}_\varepsilon(\rho)$.  Note that $\widehat{mr}_\varepsilon(\rho)$ is completely determined by
$\varepsilon$ and the spectrum of $\rho$.\smallskip

Note also that  $\,\widehat{mr}_0(\rho)=mr_0(\rho)=\rank\rho\,$ for any finite rank state $\rho$.\smallskip

\begin{example}\label{g-exam+} Continuing with Example \ref{g-exam} assume that $\rho_N$ is the Gibbs state of a quantum oscillator  corresponding to the mean number of quanta $N$ defined in (\ref{G-s}). Then
$$
\widehat{S}(\rho_N^{[m]})=q^{m}\,\frac{h(q)}{1-q}\quad\textrm{ and }\quad S(\rho_N)=\frac{h(q)}{1-q}
$$
and hence
\begin{equation*}%\label{smr+}
\widehat{mr}_\varepsilon(\rho_N)=\min\left\{m\in\N\,\left|\, q^{m}\leq \varepsilon\right.\right\}+1= [\log_q\varepsilon]+2,
\end{equation*}
where $\,[\log_q(\varepsilon)]\,$ is the integer part of the positive number $\log_q\varepsilon$.

The plots of the function  $\varepsilon\mapsto\widehat{mr}_\varepsilon(\rho_N)$ for different values of $N$
are shown on Figure 3 (in the logarithmic scales in both axis).
\end{example}

\begin{figure}[t]

\centering
\begin{center}

\includegraphics[scale=0.4, bb=400  450 500 550]{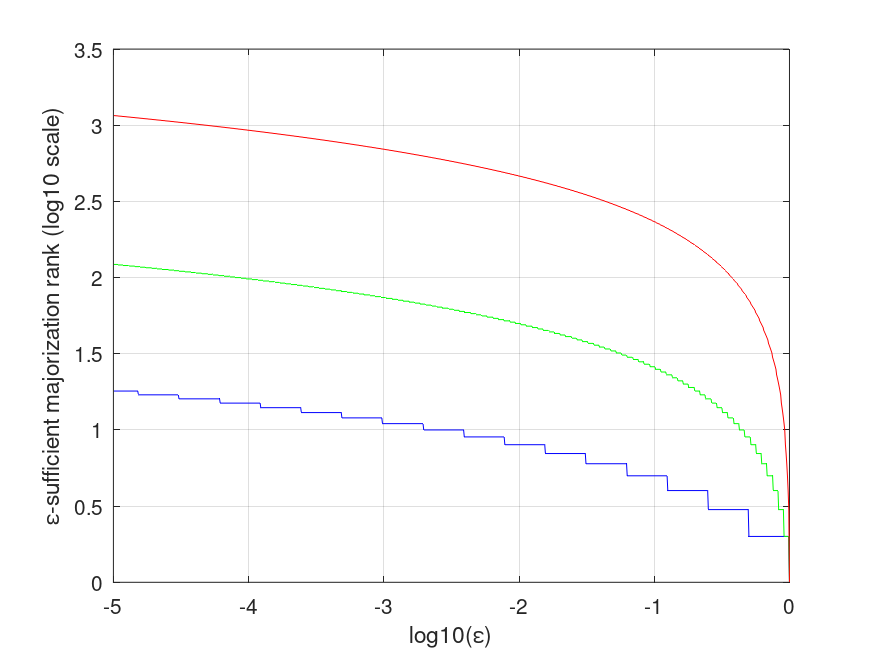}

\vspace{180pt}
\caption{The function $\,\varepsilon\mapsto\widehat{mr}_\varepsilon(\rho_N)\,$ with $\,N=1\,$ (blue line), $\,N=10\,$  (green  line) and $\,N=100\,$ (red line)
in the logarithmic scales.}
\end{center}

\label{Fig3}
\end{figure}

\section{Applications to Schur concave functions on the set of probability distributions}

All the results  of the article concerning Schur concave functions on the set of quantum states
can be easily reformulated for Schur concave functions on the set $\P^n$ of probability distributions with $\,n\leq+\infty\,$ outcomes
equipped with the \emph{total variation
distance} $\mathrm{TV}$,  which is defined for any distributions $\bar{p}=\{p_{i}\}_{i=1}^n$ and $\bar{q}=\{q_{i}\}_{i=1}^n$ in $\P^n$ as
\begin{equation*}%\label{TVD}
\mathrm{TV}(\bar{p},\bar{q})\doteq\frac{1}{2}\displaystyle\sum_{i=1}^n|p_{i}-q_{i}|.
\end{equation*}

The majorization relation $"\succ"$ for  probability distributions  is defined in (\ref{m-pd}).

Define the $m$-partial majorization relation $\bar{p}\stackrel{\,m}{\succ}\bar{q}$ for probability distributions  $\bar{p}=\{p_i\}_{i=1}^{n}$ and $\bar{q}=\{q_i\}_{i=1}^{n}$ via the system of inequalities
\begin{equation*}%\label{m-cond++}
\sum_{i=1}^{k}p^{\downarrow}_i\geq\sum_{i=1}^{k}q^{\downarrow}_i\qquad  k=1,2,..,m,
\end{equation*}
where $\{p^{\downarrow}\}_{i=1}^{n}$ and $\{q^{\downarrow}\}_{i=1}^{n}$
are the  probability distributions obtained from the distributions $\{p_i\}_{i=1}^{n}$ and $\{q_i\}_{i=1}^{n}$
by rearrangement in the non-increasing order.

If we define for  any probability distribution $\bar{p}=\{p_{i}\}_{i=1}^n$ the quantum state
 \begin{equation*}%\label{sprho}
\vartheta(\bar{p})=\sum_{i=1}^{n} p_{i}|\varphi_i\rangle\langle \varphi_i|,
\end{equation*}
where $\{\varphi_i\}_{i=1}^n$ is a fixed basis in an $n$-dimensional Hilbert space $\H$ then we
obtain a bijection from the set $\P^n$ onto the subset  $\S_0$ of $\S(\H)$ consisting of states diagonisable in the basis
$\{\varphi_i\}_{i=1}^n$. It is clear that
$$
\bar{p}\succ\bar{q}\;\Leftrightarrow\;\ \vartheta(\bar{p})\succ\vartheta(\bar{q}),\;\quad
\bar{p}\stackrel{\,m}{\succ}\bar{q}\;\Leftrightarrow\; \vartheta(\bar{p})\stackrel{\,m}{\succ}\vartheta(\bar{q})\quad\textrm{and}\quad  \mathrm{TV}(\bar{p},\bar{q})=\textstyle\frac{1}{2}\|\vartheta(\bar{p})-\vartheta(\bar{q})\|_1.
$$
It is easy to show (by checking the proofs) that we may reformulate all the results of the article by replacing the set $\S(\H)$ with the set $\S_0$.
So, the bijection $\vartheta(\bar{p})$ allows us to reformulate all the results
in terms of probability distributions and Schur concave functions
on the set $\,\P^n$ of probability distributions.

\section{Concluding remarks}

In the article, an universal technique for  obtaining upper bounds on
\begin{equation*}%\label{svp}
\sup\left\{f(\rho)-f(\sigma)\,\left|\, \textstyle\frac{1}{2}\|\rho-\sigma\|_1\leq\varepsilon \right.\right\}
\end{equation*}
and on
\begin{equation}\label{svp++}
\sup\left\{f(\rho)-f(\sigma)\,\left|\, \sigma\stackrel{\,m}{\prec}\rho,\, \textstyle\frac{1}{2}\|\rho-\sigma\|_1\leq\varepsilon \right.\right\},\qquad m\in\N,
\end{equation}
for a Schur concave function $f$ on the set of quantum states and any $\varepsilon\in[0,1]$ is proposed (Proposition \ref{pm}). Here, $\sigma\stackrel{\,m}{\prec}\rho$ means that the state $\sigma$ is $m$-partially majorized by the state $\rho$ in the sense described in Section 3. \smallskip

Then this technique was used to construct a tight upper bound on the supremum in (\ref{svp++}) depending on the spectrum of $\,\rho\,$
and  simple sufficient conditions for vanishing  this bound with $\,\min\left\{\varepsilon,\frac{1}{m}\right\}\to0\,$
have been found (Theorem \ref{main}).  \smallskip

The  proposed technique is really universal. In addition to proving Theorem \ref{main}, it can be used to prove all the semicontinuity bounds for the von Neumann entropy obtained in \cite{FCB}. It seems that this technique may be useful for quantifying continuity  of other Schur concave and Schur convex functions
on the set of quantum states. Its modification for Schur concave functions on the set probability distributions (described in Section 6) can be applied to
characteristics of discrete random variables.

In the process of solving the main tasks of the article, the state transformation $\,\rho\mapsto \rho_{m,\varepsilon}\,$ depending on $\,m\in\{0\}\cup\N\,$ and $\,\varepsilon\in[0,1]\,$  was proposed (by generalizing the construction from \cite{H&D}). It has the following  properties
$$
\rho_{m,\varepsilon}\in U_{\varepsilon}(\rho)\qquad\textrm{ and }\qquad \sigma\stackrel{\,m}{\prec}\rho\quad\Rightarrow\quad  \sigma\prec\rho_{m,\varepsilon}\qquad \forall\sigma\in U_{\varepsilon}(\rho),
$$
where $U_{\varepsilon}(\rho)$ is the $\varepsilon$-vicinity of the state $\rho$ w.r.t. the trace norm (the state $\rho_{m,\varepsilon}$ is defined at the beginning of
Section 4, the second of the above properties  is described in Remark \ref{lim+} in Section 4). This construction may be useful in analysis of any tasks, where the $m$-partial majorization relation is involved.
\bigskip

I am grateful to A.S.Holevo and to the participants of his seminar  "Quantum Probability, Statistics, Information" (the Steklov  Mathematical Institute) for useful discussion.

\end{document}